\documentclass[prx,twocolumn,superscriptaddress,longbibliography]{revtex4-2}
\pdfoutput=1

\usepackage{amsmath,braket,amssymb}
\usepackage{amsthm}
\usepackage[final]{graphicx}
\usepackage{subfigure}
\usepackage{xcolor}
\usepackage[english]{babel}
\usepackage{color}
\usepackage{mathtools}
\usepackage{physics}
\usepackage{empheq}
\usepackage{tensor}
\usepackage{quantikz}
\usepackage{hhline}
\usepackage{times}

\usepackage[colorlinks=true,citecolor=teal,linkcolor=teal,urlcolor=teal]{hyperref}

\usepackage[normalem]{ulem}
\normalfont

\usepackage[linesnumbered, ruled,vlined]{algorithm2e}

\graphicspath{{./images/},{./imagesAppendix/}}

\definecolor{forestgreen}{rgb}{0.08, 0.4, 0.13}
\definecolor{darkBlue}{rgb}{0.08, 0.13, 0.4}

\definecolor{THc}{rgb}{0.9,0.3,0.2}

\newtheorem{theorem}{Theorem}
\newtheorem{corollary}{Corollary}
\newtheorem{definition}{Definition}

\DeclareMathOperator{\SWAP}{SWAP}

\renewcommand{\eqref}[1]{Eq.~(\ref{#1})}

\begin{document}

\title{Exponential speedup in measurement property learning with post-measurement states}

\author{Zhenhuan Liu}
\email{qubithuan@gmail.com}
\thanks{equal contribution}
\noaffiliation

\author{Qi Ye}
\email{yeq22@mails.tsinghua.edu.cn}
\thanks{equal contribution}
\affiliation{Center for Quantum Information, IIIS, Tsinghua University, Beijing, China}
\affiliation{Shanghai Qi Zhi Institute, Shanghai, China}

\author{Zhenyu Cai}
\email{cai.zhenyu.physics@gmail.com }
\affiliation{Department of Engineering Science, University of Oxford, Parks Road, Oxford OX1 3PJ, United Kingdom}

\author{Jens Eisert}
\email{jense@zedat.fu-berlin.de}
\affiliation{Dahlem Center for Complex Quantum Systems, Freie Universität Berlin, 14195 Berlin, Germany}

\begin{abstract}

Learning properties of quantum states and channels is known to benefit from resources such as entangled operations, auxiliary qubits, and adaptivity, whereas the resource structure of measurement learning, namely, learning properties of quantum measurement operators, remains poorly understood. 
In this work, we identify a measurement learning task for which access limited to classical measurement outcomes leads to an exponential lower bound on the query complexity, established via a distinguishing task between a genuine quantum projective measurement and a purely classical random number generator. Remarkably, this hardness persists even when arbitrary entangled operations, auxiliary systems, and fully adaptive strategies are allowed, indicating that conventional resources for state and channel learning are ineffective in this task. 
In contrast, when access to the post-measurement quantum state is available, the same task can be solved with constant query complexity using a simple measuring-twice protocol, without requiring resources that are useful for state and channel learning.
Our results reveal post-measurement states as a qualitatively new and decisive resource for measurement learning, suggesting potential implications for the design of practical quantum certification protocols.

\end{abstract}
	
\maketitle

\section{Introduction}
Every quantum information processing task can essentially be decomposed into three stages: State initialization, quantum evolution, and measurement. 
Learning the properties of these stages underlies a wide range of tasks in quantum information processing, and hence substantial effort is dedicated to this task in various settings, let this be in notions of 
benchmarking~\cite{hashim2025benchmarking}, certification~\cite{Eisert2020certification,PRXQuantum.2.010201} and others.
Owing to the intrinsic constraints of quantum mechanics, such as the interplay of 
disturbance and information gain,
extracting information from quantum systems is subject to fundamental limits~\cite{Helstrom1976detection,Anshu2024survey}. 
Identifying the precise quantum resources that can reduce these costs is, therefore, of both foundational and practical importance in the quantum technologies. For quantum state and channel learning, extensive progress has been made in this direction, revealing resources that yield polynomial or even exponential query complexity speedups. 
Representative examples include entangled operations for estimating state purity and Pauli expectation values~\cite{Aharonov2022algorithmic,chen2021memory,huang2021demonstrating,chen2024tradeoff,Gong2024samplecomplexity,kingTriplyEfficientShadow2025}, auxiliary qubits and coherent queries for estimating channel unitarity and Pauli coefficients~\cite{chen2023unitarity,chen2022pauli,chen2024pauli,chen2025pauli}, and adaptive strategies for quantum tomography~\cite{chen2023adaptivity,Chen2024AdaptivityShadow}.  Conjugate states~\cite{king2024conjugate}, pure dilation states~\cite{Liu2025ExponentialPurification}, and inverse unitaries~\cite{tang2025amplitude} are also shown to be crucial resources in a wide range of tasks.

Measurement plays a distinctive role as the final step of any quantum information processing task, serving as the interface between quantum systems and classical data. 
In practical settings, high-quality measurement is essential for quantum randomness generation~\cite{senno2023randomness,hao2023intrinsic}, quantum error correction~\cite{deMartiiOlius2024decoding}, and quantum metrology~\cite{vittorio2004metrology}.
Despite this central importance, existing studies on measurement learning mainly focus on tomography-based protocol design, commonly referred to as “detector tomography”~\cite{luis1999measurement,fiuraifmmode2001measurement,dariano2004measurement,Lundeen2009QDT,DetectorTomography2}, and various forms of hypothesis testing~\cite{BISIO2011measurements,sedlak2014measurements,cheng2015learnability,puchala2018measurements}. 
While recent studies have successfully leveraged entangled operations to achieve polynomial speedups in detector tomography~\cite{zambrano2025fast,mele2025optimal}, the specific quantum resources required to unlock \emph{exponential speedups} in measurement learning remain largely unexplored.

\begin{figure}
\centering
\includegraphics[width=1\linewidth]{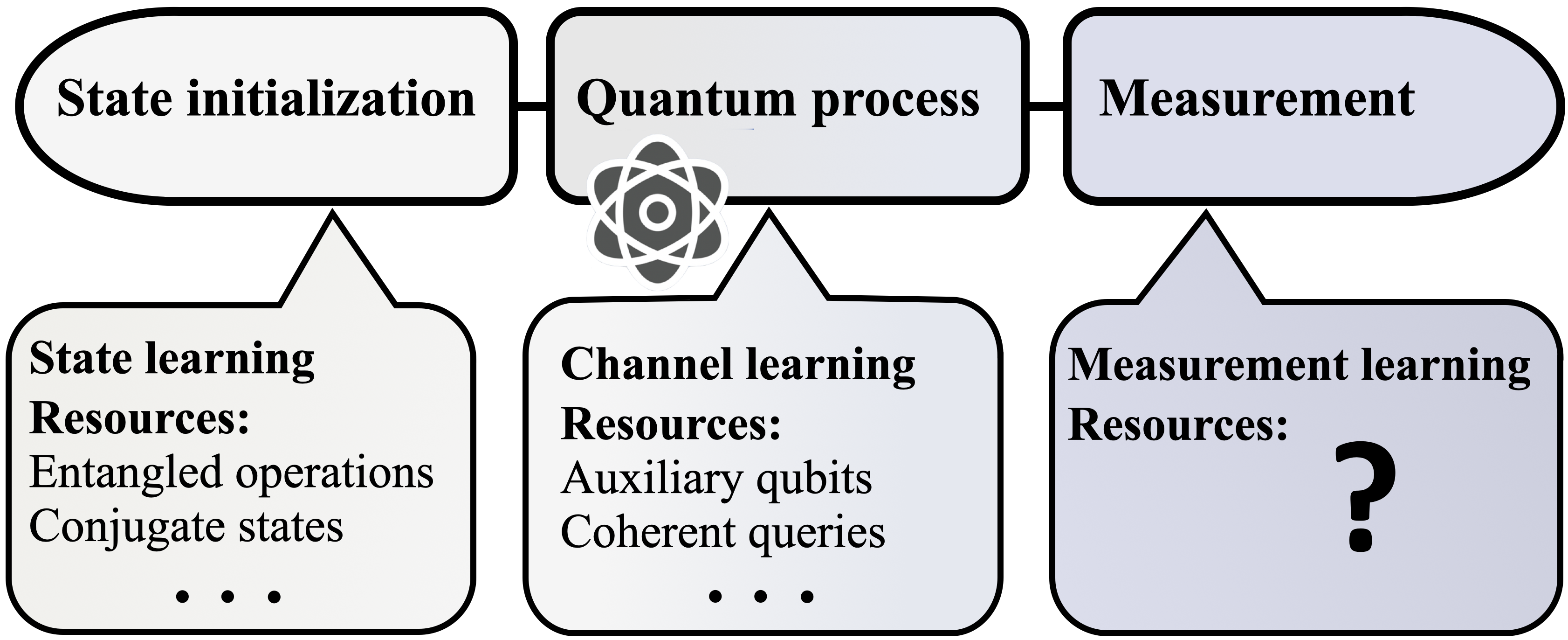}
\caption{A quantum information processing task consists of state initialization, quantum evolution, and measurement. 
Although various quantum resources have been identified that yield exponential speedups for state and channel learning, the resource structure of measurement learning remains largely unexplored.  }
\label{fig:overview}
\end{figure}

In this work, we identify a measurement property, sharpness, whose estimation query complexity exhibits an exponential separation depending on access to the post-measurement quantum state. Interestingly, this advantage cannot be replicated by conventional resources that are powerful for state and channel learning, such as entangled operations, auxiliary qubits, or adaptive strategies. 
Specifically, we prove that, with merely having access to measurement outcomes, at least $\Omega(\sqrt{d})$ queries, where $d$ denotes the system dimension, are necessary to distinguish a genuine projective measurement in a random basis from a classical random number generator. 
This hardness result holds for arbitrary quantum protocols and can be viewed as a quantum generalization of the birthday paradox. 
In contrast, when the post-measurement quantum states are also accessible, the same task can be solved with constant query complexity, independent of the system size, establishing a strict separation between the two settings. 
The protocol simply performs two identical successive measurements on the maximally mixed state, without entanglement, auxiliary systems, or adaptivity.  Taken together, our results identify a measurement learning task in which conventional resources for state and channel learning are ineffective, while post-measurement states constitute the decisive resource enabling an exponential reduction in query complexity.

\section{Measurement learning}
In practice, quantum measurements can be broadly classified into two types: \emph{destructive measurements}, which output only classical measurement outcomes, and \emph{non-destructive measurements} or quantum instruments~\cite{Davies1970instruments}, which additionally provide access to the post-measurement quantum state. 
A quantum measurement can be described with a set of measurement or Kraus operators $\{K_i\}_i$ satisfying $\sum_iK_i^\dagger K_i=\mathbb{I}$. 
When performing the measurement on a quantum state $\rho$, the probability for outputting the result $i$ is $p_i= \Tr(\rho K_i^\dagger K_i)$
and the post-measurement state is 
\begin{equation}\label{eq:post_measurement_state}
\rho_i=\frac{K_i\rho K_i^\dagger}{\Tr(K_i\rho K_i^\dagger)}.
\end{equation}
A \emph{positive operator-valued measure} (POVM) 
is a set of positive semi-definite operators $\{M_i\}_i$ satisfying $\sum_iM_i=\mathbb{I}$.
The probability of getting the result $i$ is then $\Tr(\rho M_i)$, without specifying the post-measurement state.
Given a set of measurement operators $\{K_i\}_i$, the corresponding POVM is $\{M_i=K_i^\dagger K_i\}_i$.
Given a set of POVM operators, there is no unique set of measurement operators that corresponds to it.

In this work, we focus on the measurement learning task, i.e., extracting properties of measurement operators $\{K_i\}_i$ or POVM operators $\{M_i\}_i$, in a precise sense.
In state learning tasks, we assume controllable quantum evolution and measurements.
In process learning tasks, we assume reliable state initialization and measurements, or using techniques to decouple the properties of states and measurements with the target quantum process~\cite{knill2008rb}.
Following this logic, we will assume the reliable quantum evolution and state initialization when analyzing the measurement learning tasks.

For quantum measurements, we say a measurement is ``sharp’’ if it is close to a projective measurement for which repeated measurements will only yield the same result~\cite{Buscemi2024completeoperational,saberian2024sharp}. This is what von Neumann has been referring to as a measurement of the first kind.
Conversely, a measurement is 
called ``fuzzy’’ when there is more uncertainty in the outcomes of consecutive measurements. 
Hence, we can define a quantity that faithfully characterizes how close a given POVM is from projective measurement as
\begin{equation}
\mathcal{S}(\{M_i\}_i)\coloneqq\frac{1}{d}\sum_i\Tr(M_i^2).
\end{equation}
In the remainder of this work, we will simply refer to this 
as the \emph{sharpness} of 
the measurement. 
It can be easily proved that $\mathcal{S}(\{M_i\}_i)$ reaches its maximum, $1$, if and only if the 
measurement is quantum 
projective measurements 
with $M_iM_j=M_i\delta_{i,j}$.

\section{Exponential lower bound with only classical outputs}

\begin{figure}
\centering
\includegraphics[width=1\linewidth]{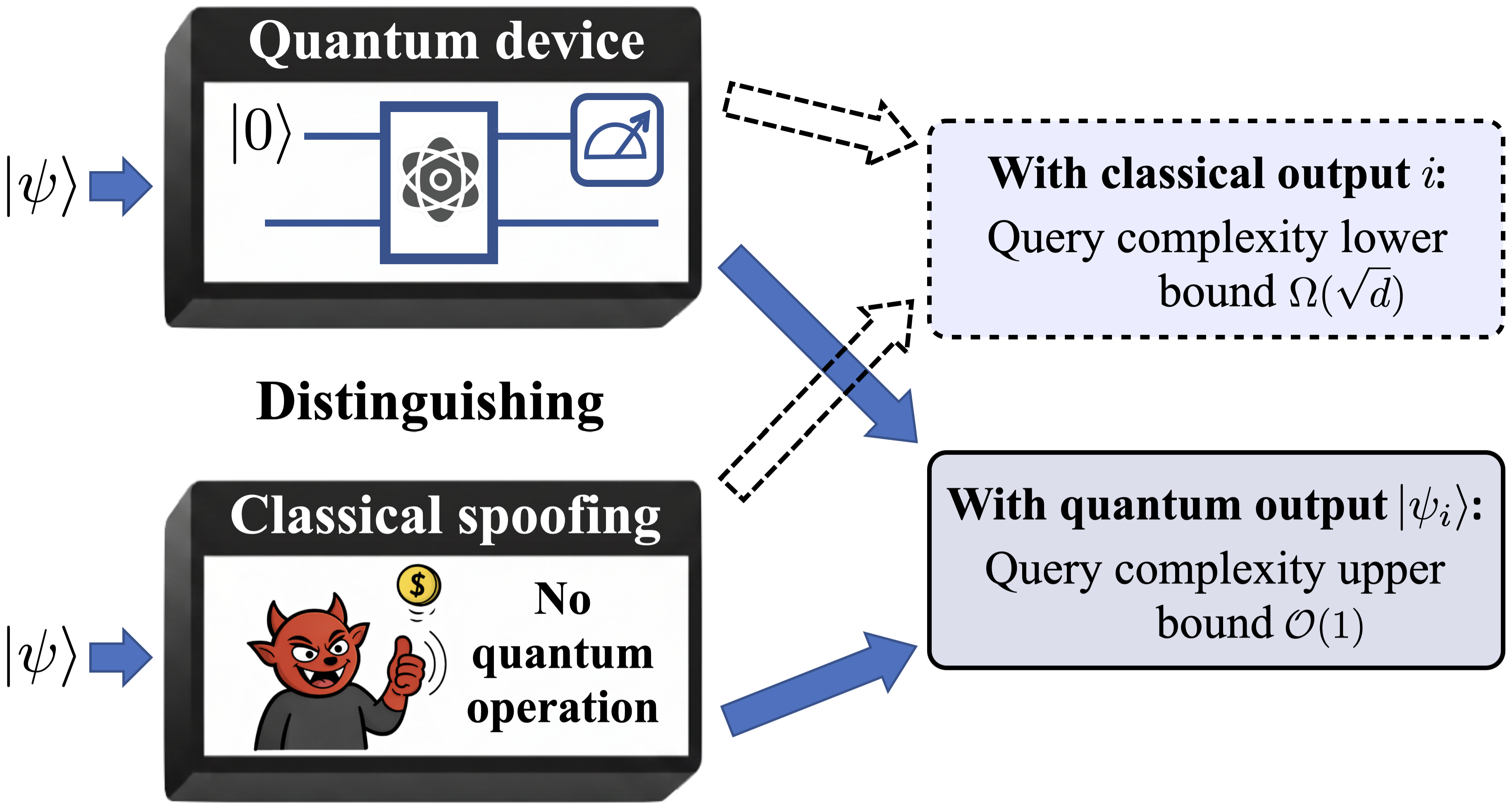}
\caption{The black box distinguishing task. }
\label{fig:bb}
\end{figure}

To derive the query complexity lower bound for estimating sharpness, we study a black-box distinguishing problem, as shown in Fig.~\ref{fig:bb}.
Here, multiple identical and independent copies of an unknown device are given and the tester can choose the input state and do further operations using the output of devices. 
The device may either faithfully perform a quantum projective measurement in an unknown and randomly chosen basis, or produce outputs sampled from a purely classical probability distribution without executing any quantum operation, precisely put as a promise problem as follows.

\begin{definition}[Black-box distinguishing task]\label{def:certification}
Given query access to a measurement device, the task is to distinguish between the following two scenarios, each occurring with equal prior probability:
\begin{enumerate}
\item[(i)] The device outputs $i \in \{0,\dots,d-1\}$ uniformly at random and leaves the input state unchanged.
\item[(ii)] The device performs a projective measurement with operators 
$\{\Pi_i = U \ketbra{i}{i} U^\dagger\}_{i=0}^{d-1}$ 
on the input state, where $U$ is an unknown $d$-dimensional Haar random unitary, outputs the measurement result $i$, and collapses the input state to the corresponding post-measurement state.
\end{enumerate}
\end{definition}

Notice that there exists a simple strategy to distinguish between these two kinds of black boxes using the birthday paradox, as shown in Fig.~\ref{fig:comb}(a).
In the uniform distribution case (i), when one samples for $\mathcal{O}(\sqrt{d})$ identical and independent times, with a constant probability, one will get a pair of identical outcomes.
If we input a fixed pure state, like 
that defined by a computational basis vector $\ket{0}$, into the black box of case (ii), the output distribution is given by the \emph{Porter–Thomas distribution}.
When sampling for $\mathcal{O}(\sqrt{d})$ times, the probability for getting two identical outcomes is almost twice as the probability of case (i), as proven in Appendix~\ref{app:upper_bound}.
Therefore, with sample complexity $\mathcal{O}(\sqrt{d})$, one can successfully distinguish these two different kinds of black 
boxes with a constant success probability even without any prior knowledge of the random unitary $U$.

\begin{figure}[htbp]
\centering
\includegraphics[width=1\linewidth]{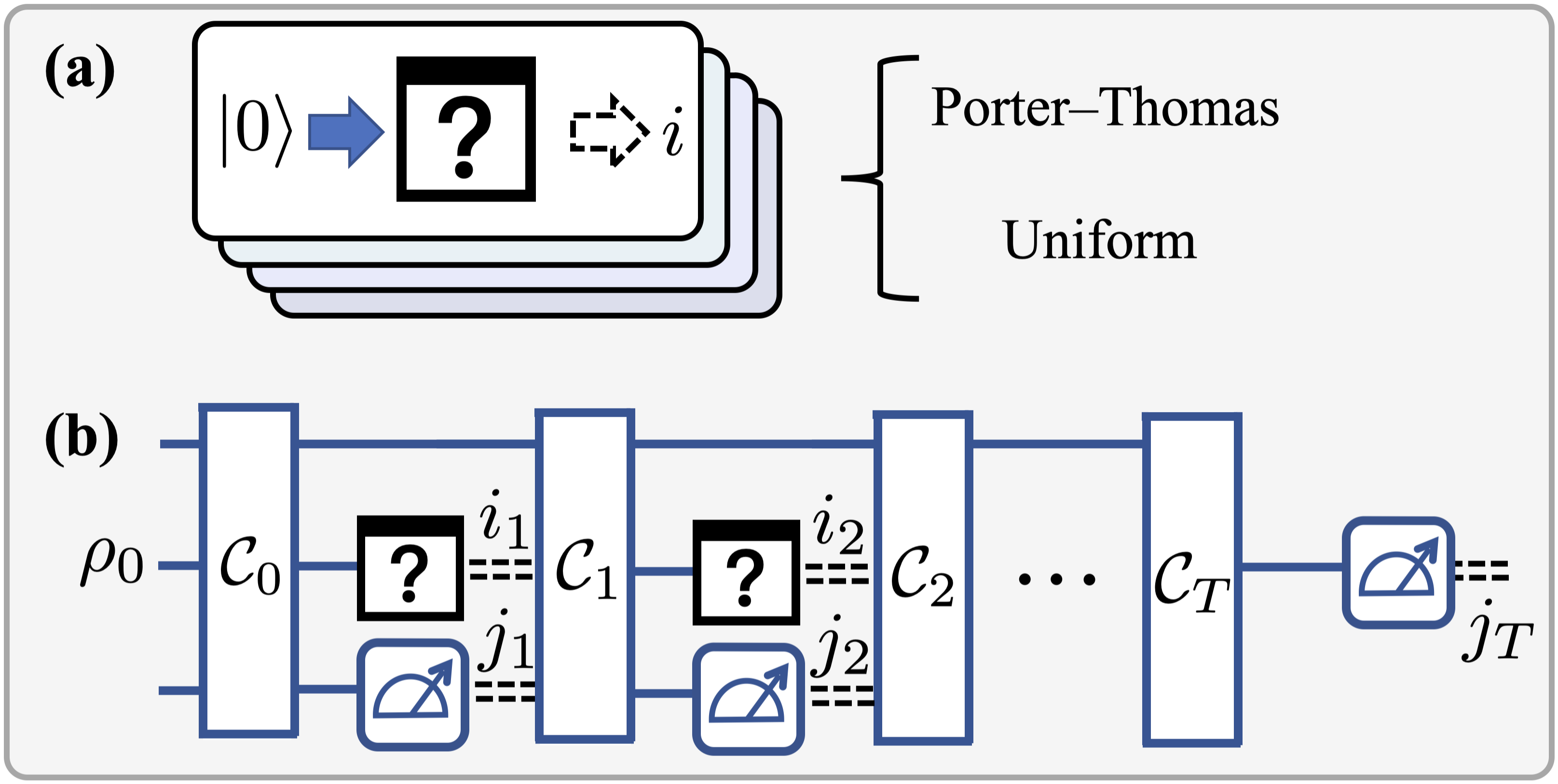}
\caption{(a) A distinguishing protocol by inputting a fixed state and comparing the output distribution. 
(b) The most general distinguishing protocol utilizing adaptivity, 
auxiliary qubits, entangled operation, and controllable measurements.
All the measurement outcomes, including the black box measurements and the controllable measurements, represented by blue boxes, will be utilized for the distinguishing task. 
}
\label{fig:comb}
\end{figure}

As mentioned before, quantum resources can speedup quantum learning tasks, such as entangled operations, coherent queries, and adaptivity.
Note that the simple protocol shown in Fig.~\ref{fig:comb}(a) does not utilize any one of these resources.
Therefore, it is natural to ask whether one can achieve a better query complexity performance with these resources.
We thus consider the most general certification protocol which allows any physical operations, including some controllable measurements in addition to the black box measurements, as shown in Fig.~\ref{fig:comb}(b).
Surprisingly, we theoretically prove that, even with the most general physical operations, the query complexity cannot break the lower 
bound of $\Omega(\sqrt{d})$. 

\begin{theorem}[Quantum birthday paradox]\label{thm:birthday}
For any quantum protocol, the query complexity for solving the task introduced in Definition~\ref{def:certification} with only classical output with constant success probability is $\Omega(\sqrt{d})$.
\end{theorem}

A direct corollary of this theorem is that, with only classical output, it is exponentially hard (in terms of query complexity) to align the outcome of random measurements.
Specifically, even with arbitrary entangled operations and adaptive strategies or ``quantum combs'' \cite{PhysRevX.15.021047},
$\Omega(\sqrt{d})$ query complexity is needed to see two identical outcomes.
We thus regard this theorem as the quantum generalization of birthday paradox.

Note that for case (i), the POVM operators are all ${\mathbb{I}}/{d}$ and thus the sharpness is simply given by ${1}/{d}$.
While for case (ii), the sharpness reaches its maximum value of $1$ and the gap between cases (i) and (ii) is a constant.
Therefore, if one can accurately estimate the measurement sharpness, one can also solve the distinguishing task.
We can thus draw the following conclusion.

\begin{corollary}[Query complexity bound]
If there exists a protocol that estimates the sharpness of any measurement to constant additive error with only classical output, the protocol must query the measurement at least $\Omega(\sqrt{d})$ times.
\end{corollary}

\section{Constant query complexity with having access to post-measurement states}

If instead of resorting to destructive measurements, we have access to non-destructive measurement apparata that allow for an access to the post-measurement state, then there exists a simple protocol to estimate the sharpness in our task without any entangled operations and adaptive strategies, as shown in Fig.~\ref{fig:twice}(a).
The protocol then amounts to simply repeating the following steps:
\begin{enumerate}
\item Start by preparing the maximally mixed state $\rho={\mathbb{I}}/{d}$.
\item Measure the state, record the result $i$ and keep the post-measurement state $\rho_i={K_iK_i^\dagger}/{{\Tr}(K_iK_i^\dagger)}$.
\item Perform the same measurement on the state $\rho_i$ to get the second measurement result $j$ and calculate the probability that $i=j$.
\end{enumerate}
According to the definition of quantum measurement, the probability we calculate is actually the unbiased estimator of the sharpness
\begin{equation}
p(i=j)=\sum_ip(i)p(j=i|i)=\mathcal{S}(\{M_i\}_i)
\end{equation}
whenever we have $[K_i,K_i^\dagger]=0$, which holds for the task defined in Definition~\ref{def:certification}.
As the value of $p(i=j)$ is contained in the interval
$[0,1]$, the variance for single-shot estimation is upper bounded by some constant.
Therefore, a constant repeating times is enough to estimate the sharpness with constant additive error.
We thus arrive at the following conclusion.

\begin{theorem}[Estimating sharpness]
There exists a protocol that estimates the sharpness with constant additive error and thus solves the task in Definition~\ref{def:certification} with constant success probability using $\mathcal{O}(1)$ queries, provided that the black box outputs both classical outcomes and post-measurement quantum states defined in \eqref{eq:post_measurement_state}, and the measurement operators are normal operators.
\end{theorem}

\begin{figure}
\centering
\includegraphics[width=1.0\linewidth]{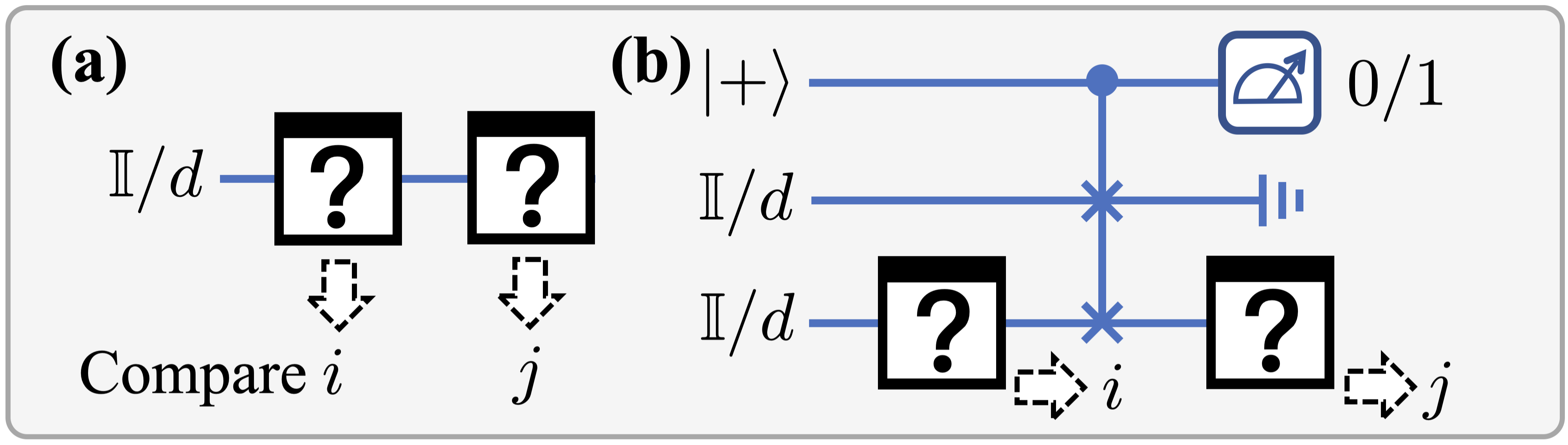}
\caption{(a) Protocol for estimating sharpness with the post-measurement state. (b) The adversarial robust protocol with one bit of quantum randomness. }
\label{fig:twice}
\end{figure}

In practical scenarios, performing two measurements on the same quantum state requires either querying the black box twice or accessing two identical copies of the device. In both situations, the simple protocol may be vulnerable to adversarial strategies, such as internal memory within a single device or classical communication between multiple devices. To enhance robustness against such behavior, we incorporate one bit of quantum randomness. Specifically, between the first and second queries, the tester measures a single qubit prepared in the state $\ket{+}$ in the computational basis and uses the random outcome to determine the input of the second query. If the outcome is $\ket{0}$, the post-measurement state obtained from the first query is sent into the black box again, reproducing the original protocol in Fig.~\ref{fig:twice}(a); if the outcome is $\ket{1}$, an independent maximally mixed state is used instead. Since the adversary does not know the measurement outcome of the control qubit, it cannot adapt its strategy to enforce either correlated or independent outputs, thereby rendering the protocol adversarially robust. This procedure is equivalent to the quantum circuit shown in Fig.~\ref{fig:twice}(b), where a controlled-SWAP gate conditionally exchanges the maximally mixed state with the post-measurement state from the first query. Consequently, for both the original ``measuring-twice'' protocol and its adversarially robust variant, the implementation overhead remains minimal, requiring only access to a maximally mixed state and at most one bit of quantum randomness.

\section{Quantumness certification}

While before we have been concerned with the abstract question of deriving query complexity bounds, we now turn to the question of how this can be used in quantumness certification, and how the 
results may offer useful insights for the design of certifiable quantum devices.
For certain classical computational problems, certification can rely on intrinsic mathematical structure. 
For instance, in Shor’s algorithm for integer factorization~\cite{shor1999}, the output can be efficiently verified classically, the problem at hand being contained in NP. However, this paradigm does not extend to intrinsically quantum tasks such as quantum many-body simulation, where reliable certification would in general require trusted quantum operations of comparable complexity. 
Moreover, for non-computational tasks such as quantum random number generation~\cite{herrero2017qrng}, purely classical devices may mimic quantum output statistics, rendering certification based solely on classical outcomes fundamentally difficult~\cite{Pironio2010randomness,Mal2016randomness}. 

In digital quantum computing platforms, measurements are typically implemented as a unitary evolution $U$ followed by a computational basis measurement. 
If measurements are restricted to a fixed and known basis, such as the computational basis, then even with access to post-measurement states, the distinguishing task remains exponentially hard by virtue of Theorem~\ref{thm:birthday}: Obviously, having access to the post-measurement state does not provide any additional information. However, if the device provider additionally supplies the inverse unitary $U^\dagger$ applied to the post-measurement state before releasing it, the client can implement the simple ``measuring twice'' protocol using only a maximally mixed state and a single bit of quantum randomness. 
In this way, quantumness can be certified with constant query complexity and minimal quantum capability on the verifier’s side. 
This demonstrates how the resource identified in this work can be operationally leveraged in practical certification settings.

\section{Discussion}

At this point, we discuss the broader implications of our technical results.
As stressed before, state preparation and measurement constitute the first and final stages of any quantum information processing task and exhibit a natural duality: For this reason, detector and state tomography can be seen as being dual to each other \cite{Lundeen2009QDT,chen2023adaptivity,zambrano2025fast,mele2025optimal}. Only on rare occasions, one can acquire knowledge about both at the same time \cite{Blind}. 
It is, therefore, tempting to expect that similar quantum resources would be useful for both state learning and measurement learning, for instance, entangled operations for purity estimation. However, our results reveal a striking contrast and asymmetry in those tasks: There exist measurement learning tasks for which entangled operations provide no advantage when only classical measurement outcomes are available. 
In contrast, access to post-measurement quantum states enables an exponential query complexity speedup of order $\sqrt{d}$ through repeated measurements, paralleling the role of entangled operations in state 
purity estimation~\cite{chen2021memory}.

To add yet a further perspective, 
a quantum measurement can also be viewed as a special type of quantum channel. Without access to 
post-measurement states, the 
corresponding measurement channel is 
\begin{equation}
\mathcal{M}_{\mathrm{c}}(\rho)=\sum_i \ketbra{i}{i}_{\mathrm{c}} \Tr(K_i \rho K_i^\dagger),
\end{equation}
where $\mathrm{c}$ denotes the classical register storing the outcome. 
With post-measurement states being available, the measurement channel becomes 
\begin{equation}
\mathcal{M}_{\mathrm{q}}(\rho)=\sum_i \ketbra{i}{i}_{\mathrm{c}} \otimes K_i \rho K_i^\dagger.
\end{equation}
Our results demonstrate that access to $\mathcal{M}_{\mathrm{q}}$ yields an exponential reduction in query complexity compared with access to $\mathcal{M}_{\mathrm{c}}$ for estimating sharpness. 
Although sharpness is formally reminiscent of channel unitarity, defined as the purity of the Choi state, it does not coincide with the unitarity of either $\mathcal{M}_{\mathrm{q}}$ or $\mathcal{M}_{\mathrm{c}}$. 
Therefore, the speedup in sharpness estimation represents a qualitatively new form of learning advantage, distinct from previously identified advantages in unitarity learning~\cite{chen2023unitarity}.

Indeed, our work connects to several prior lines of research. A direct corollary of Theorem~\ref{thm:birthday} establishes the hardness of a distributional hypothesis testing problem, namely distinguishing between the uniform distribution and the Porter--Thomas distribution generated by an unknown random unitary. Similar questions arise in the certification of \emph{quantum random sampling schemes}~\cite{black-box-verification,Boixo}, where the output distribution is often well approximated by a Porter--Thomas distribution~\cite{hangleiter2023sampling}, including IQP circuits~\cite{shepherd_temporally_2009}, boson sampling~\cite{BosonSampling}, random circuit sampling~\cite{Boixo}, and certain quantum simulation schemes~\cite{NewSupremacy}. 
It is known that when the random unitary is known in advance, the Porter-Thomas distribution can be distinguished from the uniform distribution with constant query complexity~\cite{aaronson2013bosonsampling}, whereas without such prior knowledge the task becomes exponentially hard~\cite{SampleComplexity}. 
By contrast, we show that when the two distributions originate from physically distinct mechanisms, one produced by a genuine quantum measurement and the other by a purely classical process, access to post-measurement quantum states enables an exponential speedup in query complexity, even without knowledge of the random unitary. 
More generally, related distribution distinguishing problems have been studied under both classical and quantum query models~\cite{gilyen2019distributional,wang2024distribution}, where at most polynomial quantum speedups are known, highlighting that physically meaningful quantum access can fundamentally alter the complexity landscape.

As discussed in the previous section, outputting the post-measurement quantum state in a digital architecture requires applying an additional inverse unitary $U^\dagger$ after the computational-basis measurement. 
From this perspective, the 
exponential speedup can be traced back to the availability of this inverse unitary as oracle access. 
Previous work has shown that access to such inverse unitaries can yield quadratic improvements in tasks such as amplitude amplification 
and amplitude estimation~\cite{tang2025amplitude}. 
In our setting, however, 
each unitary $U$ is 
immediately followed by 
a computational-basis measurement, which destroys coherent access and fundamentally limits the power of standard quantum query strategies. 
We conjecture that it is precisely this measurement-induced collapse that allows post-measurement state access to produce an exponential, rather than merely quadratic, advantage.

\begin{acknowledgements}
We acknowledge the insightful discussion with Yunchao Liu, Alexander Nietner, Yifan Tang and Weiyuan Gong and acknowledge Antonio Anna Mele for introducing Refs.~\cite{zambrano2025fast,mele2025optimal} to this context. 
Q.Y. acknowledges support from the National Natural Science Foundation of China (No. T24B2002).
Z.C. acknowledges support from EPSRC projects Robust and Reliable Quantum Computing (RoaRQ, EP/W032635/1) and the EPSRC Quantum
Technologies Career Acceleration Fellowship (UKRI1226).
E.J. acknowledges support by the BMFTR (PraktiQOM, QuSol, HYBRID++, PQ-CCA),
the Munich Quantum Valley, Berlin Quantum, the Quantum Flagship (Millenion, PasQuans2), the DFG (CRC 183, SPP 2514), the Clusters of Excellence (MATH+, ML4Q), and the European Research Council (DebuQC).

\emph{Note added.}
While finalizing this work, we noticed an excellent concurrent paper by Eid and Quintino~\cite{eid2026postmeasurementstates} that investigates closely related topics. 
The two works have been carried out independently. 
Key differences include the following. 
Eid and Quintino~\cite{eid2026postmeasurementstates} explore the advantages of post-measurement states in single-shot discrimination of two different measurement settings, focusing on the success probability. 
In contrast, (i) we study the task of estimating certain properties of measurements, (ii) we primarily focus on query complexity, demonstrating an exponential separation in query complexity enabled by post-measurement states, and (iii) we go beyond the single-shot protocol and prove that the advantage of post-measurement states holds for arbitrary discrimination protocols.

\end{acknowledgements}

\onecolumngrid
\appendix
\clearpage

\section{Distinguishing between uniform and Porter-Thomas distributions}\label{app:upper_bound}

This appendix provides details of the arguments of the main text. In this section, we would like to prove that $\mathcal{O}(\sqrt{d})$ query complexity is enough to distinguish between the uniform distribution and the distribution defined by 
\begin{equation}\label{eq:PT}
p_U(x)=\abs{\bra{x}U\ket{0}}^2,
\end{equation}
where $U$ is sampled according to the $d$-dimensional Haar measure. It is known that for a fixed $x$, the probability $p_U(x)$ approximately follows the Porter-Thomas distribution when $d$ is large. 
Note that the subsequent theorem has some resemblance in mindset with the results of Ref.\ \cite{SampleComplexity} in which the output distribution of a boson sampler is tested against the uniform distribution. The following 
theorem covers this 
binary hypothesis testing problem between two probability distributions over the same sample space.

\begin{theorem}[Sample complexity of 
discriminating the Porter-Thomas from the uniform distribution]
Let $\mathcal{H}_0$ denote the hypothesis that samples are drawn from the $d$-dimensional uniform distribution, and let $\mathcal{H}_1$ denote the hypothesis that samples are drawn from $p_U(x)$ for an unknown Haar random unitary $U$.
Then there exists a test based on $N = \mathcal{O}(\sqrt{d})$ independent samples whose success probability is at least $2/3$.
\end{theorem}

\begin{proof}
    Given any probability distribution $p(x)$, suppose that we sample according to this distribution for $N$ times and obtain results $\{x_a\}_{a=1}^N$. We calculate a random 
    variable $C$, called the collision 
    counter, as
    \begin{equation}
        C = \sum_{a<b}\mathbf{1}[x_a=x_b].
    \end{equation}
    We now calculate the expectation and 
    variance of $C$ under the two hypotheses. 
    For convenience, denote
    \begin{equation}
        s_2=\sum_xp(x)^2 , \ \ s_3=\sum_xp(x)^3.
    \end{equation}
    It is easy to see that $\mathbb{E}[C] = \binom{N}{2} \mathrm{Pr}[x_1=x_2] = \binom{N}{2} s_2$. 
    \begin{equation}
        \mathrm{Var}[C] = \mathbb{E}[C^2]-\mathbb{E}[C]^2 = \sum_{a<b, c<d}(\mathrm{Pr}[x_a=x_b, x_c=x_d]-\mathrm{Pr}[x_a=x_b]\mathrm{Pr}[x_a=x_b]).
    \end{equation}
    The inner term gives the value $0$ when $\{a, b\}\cap \{c, d\}=\emptyset$, $s_3-s_2^2$ when $\left|\{a, b\}\cap \{c, d\}\right|=1$, and $s_2-s_2^2$ when $(a, b)=(c, d)$. Therefore, 
    \begin{equation}
        \mathrm{Var}[C] = \binom{N}{2}(s_2-s_2^2)+6\binom{N}{3}(s_3-s_2^2).
    \end{equation}
    For hypothesis $\mathcal{H}_0$, $s_2 = \frac{1}{d}, s_3=\frac{1}{d^2}$, so 
    \begin{equation}
        \mathbb{E}[C] = \binom{N}{2}\frac{1}{d}, \ \ \mathrm{Var}[C] = \binom{N}{2}\left(\frac{1}{d}-\frac{1}{d^2}\right)\leq \binom{N}{2}\frac{1}{d}.
    \end{equation}
    In the case of hypothesis $\mathcal{H}_1$, we additionally need to consider the randomness from $U$. According to the Haar integral, we find 
    \begin{equation}
        \mathbb{E}_U[s_2] = \frac{2}{d+1}, \mathbb{E}_U[s_3] = \frac{6}{(d+1)(d+2)}.
        \end{equation}Therefore,
    \begin{equation}
        \mathbb{E}[C] = \binom{N}{2}\frac{2}{d+1}, \ \ \mathrm{Var}[C] \leq \binom{N}{2}\frac{2}{d} + \binom{N}{3}\frac{36}{d^2}.
    \end{equation}
    Taking $N = 20\sqrt{d}$, the expectations under two hypotheses differ by at least $\binom{N}{2}\frac{2}{d+1}-\binom{N}{2}\frac{1}{d}\ge 100$, while the variances under two hypotheses are both at most 400. By Chebyshev's inequality, with probability at least $1-\frac{400}{50^2}>2/3$, we can estimate $\mathbb{E}[C]$ within additive error $50$, which is sufficient to distinguish the two hypotheses.
\end{proof}

\section{General query complexity lower bound with only classical output}
\newcommand{\memory}{\text{\textsf{MEM}}}
\newcommand{\inp}{\text{\textsf{IN}}}
\newcommand{\inputs}{\text{\textsf{INs}}}
\newcommand{\outp}{\text{\textsf{OUT}}}
\newcommand{\outputs}{\text{\textsf{OUTs}}}
\newcommand{\total}{\text{\textsf{TOTAL}}}
\renewcommand{\read}{\text{\textsf{READ}}}
\newcommand{\initial}{\text{int}}
\newcommand{\TV}{\text{TV}}
\newcommand{\bx}{\mathbf{x}}
\newcommand{\cM}{\mathcal{M}}
\newcommand{\cC}{\mathcal{C}}
\newcommand{\cV}{\mathcal{V}}
Here, we prove that any protocol requires at least $\sqrt{d}$ queries to distinguish between a truly quantum measurement $\cM_U=\{M_x = U\ketbra{x} U^\dagger\}_x$ and a classical uniform random number generator, which corresponds to the trivial POVM 
\begin{equation}
\cM_\mathrm{c}=\left\{M_x = \frac{\mathbb{I}}{d}\right\}_x.
\end{equation}
We consider the protocol at hand in the following form. It first prepares an $(m+n)$-qubit initial state $\rho_0$, where $m$ is the number of auxiliary qubits and can be arbitrarily large. 
In round $t$, it measures the last $n$ qubits of $\rho_{t-1}$ using the unknown measurement $\cM$ (which is either $\cM_U$ or $\cM_\mathrm{c}$), obtaining outcome $x_t \in [d]$, and replaces the measured qubits with the computational basis state vector $\ket{x_t}$. 
It then performs a unitary $V_t$ to obtain $\rho_t$ and proceeds to the next round. 
After $T$ rounds of such queries, it performs a final computational basis measurement on all qubits and obtains an outcome $z$. 
Based on $\bx=(x_1, \ldots, x_T)$ and $z$, the protocol decides whether $\cM=\cM_U$ or $\cM=\cM_\mathrm{c}$. 
Here, we depict the circuit with 3 queries. 

\begin{center}
\includegraphics{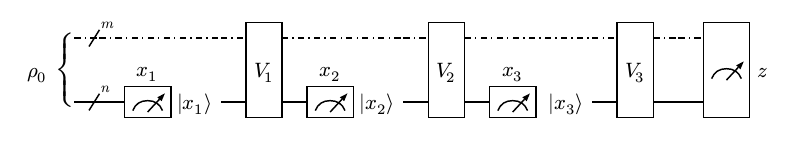}
\end{center}

We have made the following simplifications without loss of generality:
\begin{itemize}
\item The distinguisher does not perform mid-circuit measurements other than queries to the unknown measurement, as all such mid-circuit measurements can be postponed to the end with the price of more auxiliary qubits and entangled gates.

\item We assume that the outcome $x$ is fed back to the circuit as the computational state vector $\ket{x}$. 
In general, there could be a powerful classical computer that decides on the next input state. 
However, such classical computation can always be implemented within the subsequent quantum circuit.
\item We assume that the final measurement is a computational basis measurement. 
Since we allow arbitrary quantum operations before the final measurement as well as an arbitrary number of ancilla qubits, any general measurement can be implemented in this way.
\end{itemize}
Therefore, our model captures the most general distinguisher. 
We now prove our main theorem, which we restate formally as follows.

\begin{theorem}[Query complexity bound without access to post-measurement states]
    Given access to an unknown measurement $\cM$ without post-measurement states, any protocol requires $\Omega(\sqrt{d})$ queries to distinguishing between the following two cases with probability at least $2/3$:
    \begin{itemize}
        \item $\cM=\cM_\mathrm{c}$,
        \item $\cM=\cM_U$, where $U$ is a Haar random unitary.
    \end{itemize}
\end{theorem}
\begin{proof}
    We calculate the joint distribution of the outcomes $(\bx, z)$ when the unknown measurement is $\cM = \{M_x\}_{x \in [d]}$. 
    Our goal is to show that when $T^2 = o(d)$, the total variation distance between the outcome distributions in the two cases is small, making them indistinguishable.

    For convenience, we refer to the first $m$ qubits as register $\memory$ and the last $n$ qubits as register $\read$. 
    For $x, y \in [d]$, we define a completely positive (CP) map $\cM_{x, y}$ from $\memory \otimes \read$ to $\memory \otimes \read$ as 
    \begin{equation}
    \cM_{x, y}(\rho) = \Tr_{\read}(\rho M_{x, \read}) \otimes \ketbra{y}_{\read}.
    \end{equation}
    In other words, $\cM_{x,y} (\rho)$ represents the process of measuring the $\read$ register of $\rho$ using $\cM$, postselecting on outcome $x$, and replacing the $\read$ register with $\ketbra{y}$. 
    Note that $\cM_{x,y}$ is not trace-preserving; in fact, $\Tr(\cM_{x,y}(\rho))$ gives the probability of the postselection event. 
    Let $\cV_t (\cdot)=V_t \cdot V_t^\dagger$. 
    By definition, the probability of obtaining outcome $(\bx, z)$ is 
    \begin{equation} \label{eq:probability}
    p^\cM (\bx, z) = \bra{z}\cV_T \circ \cM_{x_T, x_T} \circ \cV_{T-1} \circ \cM_{x_{T-1}, x_{T-1}} \circ  \cdots \circ \cV_1 \circ \cM_{x_1, x_1}(\rho_0)\ket{z}.
    \end{equation}

    This formula characterizes the outcome distribution of the adaptive protocol. 
    We now relate it to the outcome distribution of a parallel protocol with projection, following a similar approach used in Ref.~\cite{thomas2025extreme}. 
    For each $\cM_{x, y}$, we introduce two ancillary $d$-dimensional registers, $\inp$ and $\outp$. 
    Let $\ket{\Psi}_{\read,\outp}=\sum_{x \in [d]} \ket{x , x}$ denote the (unnormalized) Bell state vector on $\read$ and $\outp$, and let $\SWAP_{\read, \inp}$ denote the swap operator between registers $\read$ and $\inp$. 
    We can then rewrite
    \begin{align}
      \Tr_\read (\rho M_{x,\read})&=\Tr_{\inp} (\SWAP_{\read,\inp} \rho \SWAP_{\read,\inp} M_{x,\inp}),\\
      \ketbra{y}_{\read} &= \bra{y}_{\outp}\ketbra{\Psi}_{\read, \inp} \ket{y}_{\outp}.
    \end{align}
    Therefore, we obtain 
    \begin{equation}\label{eq:choi}
    \cM_{x,y}(\rho)=\Tr_{\inp,\outp}(\mathcal{C}(\rho) \cdot (M_{x,\inp} \otimes \ketbra{y}_{\outp})),
    \end{equation}
    where 
    \begin{equation}
      \mathcal{C}(\rho) = \SWAP_{\read,\inp} \rho \SWAP_{\read,\inp} \otimes \ketbra{\Psi}_{\read, \outp}
    \end{equation}
    is a \emph{completely positive trace-preserving} (CPTP) map from $\memory \otimes \read$ to $\memory \otimes \read \otimes \inp \otimes \outp$. Note that $\mathcal{C}$ is independent of $\cM, x, y$. 
    
    At this point, we
  substitute all occurrences of $\cM_{x_t, x_t}$ in \eqref{eq:probability} with the expression from \eqref{eq:choi}. 
    We denote the $\inp$ and $\outp$ registers for the $t$-th query by $\inp_t$ and $\outp_t$, respectively, and denote $\memory \otimes \read$ by $\total$. 
    We use $\inputs$ and $\outputs$ to denote the collections of all $\inp_t$ and $\outp_t$, respectively. 
    This step yields an operator on $\total \otimes \inputs \otimes \outputs$ given by
    \begin{equation}
      A = \cV_T \circ \mathcal{C} \circ \cV_{T-1} \circ \mathcal{C} \circ \cdots \circ \cV_1 \circ \mathcal{C} (\rho_0).
    \end{equation}
    Since both $\mathcal{C}$ and $\cV$ are CPTP maps, $A$ is a density matrix. 
    We can now express the probability in \eqref{eq:probability} as an inner product:
    \begin{equation}
      p^\cM (\bx, z) = \Tr(A \cdot (\ketbra{z}_\total \otimes M_{\bx,\inputs} \otimes \ketbra{\bx}_{\outputs})),
    \end{equation}
    where $M_{\bx,\inputs}=\bigotimes_{t=1}^T M_{x_t,\inp_t}$. 
    Let $A^{\bx, z}=\bra{z\,\bx}_{\total,\outputs} A \ket{z\,\bx}_{\total, \outputs}$ denote the resulting operator on $\inputs$. 
    Then $p^\cM (\bx, z)=\Tr(A^{\bx, z} M_{\bx,\inputs})$.  
    With this inner product formulation, we are ready to bound the total variation distance between the outcome distributions in the two cases. For this, we find
    \begin{equation}\label{eq:tv-distance1}
    \begin{aligned}
      d_{\TV}(p^{\cM_\mathrm{c}}, \mathbb{E}_U p^{\cM_U})&=\frac{1}{2} \sum_{\bx, z}\left|p^{\cM_\mathrm{c}}(\bx, z) - \mathbb{E}_U p^{\cM_U}(\bx, z)\right|\\
      &=\frac{1}{2} \sum_{\bx, z}\left|\frac{1}{d^T} \Tr(A^{\bx, z})- \mathbb{E}_U \bra{\bx}U^{\otimes T} A^{\bx, z} U^{\dagger \otimes T} \ket{\bx}\right|.
    \end{aligned}
    \end{equation}
    Here, $\mathbb{E}_U U^{\otimes T} A^{\bx, z} U^{\dagger \otimes T}$ represents the Haar twirling of $A^{\bx, z}$.  
    The result of Haar twirling can be expressed as $\sum_{\tau, \pi \in S_T} \Tr(A \tau^{-1}) \text{Wg}_{\tau,\pi}\pi$, where $S_T$ is the symmetric group on $T$ elements, $\pi$ and $\tau$ are permutations (with notation overloaded to also denote their representations on the $T$-fold tensor space), and $\text{Wg}$ is a $T! \times T!$ matrix called the Weingarten matrix. 
    It is known that when $T^2\le d$, $\text{Wg}$ is close to the identity up to a factor of $1/d^T$. 
    Specifically, it has been shown in Ref.~\cite{thomas2025extreme} that $\|d^T\text{Wg}- I_{T!}\|_\infty \le T^2/d$ when $T^2\le d$. 
    Therefore,
    \begin{align}
      &\left\|\mathbb{E}_U U^{\otimes T} A^{\bx, z} U^{\dagger \otimes T}-\sum_{\pi}\Tr(A^{\bx, z} \pi^{-1})\frac{1}{d^T} \pi\right\|_\infty\\
      \nonumber
      =& \left\|\sum_{\tau, \pi} \Tr(A^{\bx, z} \tau^{-1})(\text{Wg}_{\tau,\pi}-\delta_{\tau,\pi}/d^T) \pi\right\|_\infty\\
       \nonumber
      \le&\sum_{\tau, \pi} \left|\Tr(A^{\bx, z} \tau^{-1})\right| \left|\text{Wg}_{\tau,\pi}-\delta_{\tau,\pi}/d^T\right| \|\pi\|_\infty\\
       \nonumber
      \le& \frac{1}{d^T} \sum_{\tau, \pi}\Tr(A^{\bx, z}) \frac{T^2}{d} 
      \\
       \nonumber
      =& \Tr(A^{\bx, z}) \frac{T^2}{d} \frac{(T!)^2}{d^T} \le \Tr(A^{\bx, z}) \frac{T^2}{d} \left(\frac{T^2}{d}\right)^T \le \Tr(A^{\bx, z}) \frac{T^2}{d}.
       \nonumber
    \end{align}
    Substituting this bound into \eqref{eq:tv-distance1}, we obtain
    \begin{align}
      d_\TV (p^{\cM_\mathrm{c}}, \mathbb{E}_U p^{\cM_U})\le&\frac{1}{2} \sum_{\bx, z} \left|\frac{1}{d^T} \Tr(A^{\bx, z}) - \frac{1}{d^T} \sum_{\pi} \Tr(A^{\bx, z} \pi^{-1})\bra{\bx}\pi\ket{\bx}\right| + \frac{1}{2} \sum_{\bx, z} \Tr(A^{\bx, z}) \frac{T^2}{d}\\
       \nonumber
      =&\frac{T^2}{2d} + \frac{1}{2} \sum_{\bx, z} \frac{1}{d^T}\left|\sum_{\pi \ne e} \Tr(A^{\bx, z}\pi^{-1}) \bra{\bx}\pi\ket{\bx}\right|\\
       \nonumber
      \le& \frac{T^2}{2d} + \frac{1}{2} \sum_{\bx, z, \pi \ne e} \frac{1}{d^T} \Tr(A^{\bx, z})\bra{\bx}\pi\ket{\bx}\\
       \nonumber
      =& \frac{T^2}{2d} + \frac{1}{2} \sum_{\bx, \pi \ne e} \frac{1}{d^T} \Tr(\bra{\bx}_{\outputs} A \ket{\bx}_{\outputs}) \bra{\bx}\pi\ket{\bx}\\
       \nonumber
      \le& \frac{T^2}{2d} + \frac{1}{2} \sum_{\bx, \pi \ne e} \frac{1}{d^T} \bra{\bx}\pi\ket{\bx}\\
       \nonumber
      =& \frac{T^2}{2d} + \frac{1}{2} \sum_{\pi \ne e} d^{c(\pi)-T}.
    \end{align}
    In the last line, $c(\pi)$ 
    denotes the number of cycles in 
    the permutation $\pi$. 
    Since $\bra{\bx}\pi\ket{\bx}$ 
    is nonzero only when $\bx$ 
    is constant on each cycle of $\pi$, there are exactly $d^{c(\pi)}$ 
    such vectors $\bx$. 
    Thus, $\sum_{\bx}\bra{\bx}\pi\ket{\bx} = d^{c(\pi)}$. 
    It is known that $\sum_\pi d^{c(\pi)-T}=(1+1/d)(1+2/d)\cdots (1+(T-1)/d)$ (cf.\ Stirling numbers of the first kind). 
    This can be upper bounded by 
    $\exp(T^2/d)$, which is at most $1+2T^2/d$ when $T^2\le d$. 
    Therefore, we have
    \begin{equation}
      d_\TV (p^{\cM_\mathrm{c}}, \mathbb{E}_U p^{\cM_U}) \le \frac{T^2}{2d} + \frac{T^2}{d} = \frac{3T^2}{2d}.
    \end{equation}
    Since the success probability of any algorithm for distinguishing between two distributions is upper bounded by $1/2 + d_{\TV}/2$, we conclude that when $T \le \sqrt{d}/3$, no distinguisher can achieve a success probability greater than $2/3$, completing the proof.
\end{proof}
\end{document}